\documentclass[12pt,english]{article}
\usepackage[T1]{fontenc}
\usepackage[latin9]{inputenc}
\usepackage{color}
\usepackage{babel}
\usepackage{float}
\usepackage{booktabs}
\usepackage{bm}
\usepackage{amsmath}
\usepackage{amsthm}
\usepackage{amssymb}
\usepackage{graphicx}
\usepackage[authoryear]{natbib}
\usepackage[unicode=true,pdfusetitle,
 bookmarks=true,bookmarksnumbered=false,bookmarksopen=false,
 breaklinks=false,pdfborder={0 0 1},backref=false,colorlinks=true]
 {hyperref}
\hypersetup{
 citecolor = blue}

\makeatletter

\providecommand{\tabularnewline}{\\}
\floatstyle{ruled}
\newfloat{algorithm}{tbp}{loa}
\providecommand{\algorithmname}{Algorithm}
\floatname{algorithm}{\protect\algorithmname}

\theoremstyle{plain}
    \ifx\thechapter\undefined
      \newtheorem{prop}{\protect\propositionname}
    \else
      \newtheorem{prop}{\protect\propositionname}[chapter]
    \fi
\theoremstyle{remark}
    \ifx\thechapter\undefined
      \newtheorem{rem}{\protect\remarkname}
    \else
      \newtheorem{rem}{\protect\remarkname}[chapter]
    \fi
\theoremstyle{plain}
    \ifx\thechapter\undefined
	    \newtheorem{thm}{\protect\theoremname}
	  \else
      \newtheorem{thm}{\protect\theoremname}[chapter]
    \fi
\theoremstyle{plain}
    \ifx\thechapter\undefined
      \newtheorem{lem}{\protect\lemmaname}
    \else
      \newtheorem{lem}{\protect\lemmaname}[chapter]
    \fi

\usepackage[margin=1in]{geometry}
\usepackage{setspace}
\usepackage{authblk}

\title{Adaptive Conditional Distribution Estimation with Bayesian Decision Tree Ensembles}
\author[1]{Yinpu Li}
\author[2,*]{Antonio R. Linero}
\author[2]{Jared Murray}
\affil[1]{Department of Statistics, Florida State University}
\affil[2]{Department of Statistics and Data Sciences, University of Texas at Austin}
\affil[*]{\normalfont\texttt{antonio.linero@austin.utexas.edu}}

\date{\small\textbf{An updated and corrected version of this manuscript will appear in
    Journal of the American Statistical Association}}

\makeatother

\providecommand{\lemmaname}{Lemma}
\providecommand{\propositionname}{Proposition}
\providecommand{\remarkname}{Remark}
\providecommand{\theoremname}{Theorem}

\begin{document}
\maketitle 
\global\long\def\E{\mathbb{E}}%
\global\long\def\Var{\textnormal{Var}}%
\global\long\def\Tree{\mathcal{T}}%
\global\long\def\sM{\mathcal{M}}%
\global\long\def\Data{\mathcal{D}}%
\global\long\def\Sieve{\mathcal{F}}%
\global\long\def\Ancestors{\mathcal{A}}%
\global\long\def\Basis{\mathcal{B}}%
\global\long\def\Reals{\mathbb{R}}%
\global\long\def\iid{\stackrel{\textnormal{iid}}{\sim}}%
\global\long\def\yt{\widetilde{y}}%
\global\long\def\BART{\textnormal{BART}}%
\global\long\def\SBART{\textnormal{SBART}}%
\global\long\def\Normal{\textnormal{Normal}}%
\global\long\def\Beta{\textnormal{Beta}}%
\global\long\def\Bernoulli{\textnormal{Bernoulli}}%
\global\long\def\Gam{\textnormal{Gam}}%
\global\long\def\Uniform{\textnormal{Uniform}}%
\global\long\def\Dirichlet{\textnormal{Dirichlet}}%
\global\long\def\Categorical{\textnormal{Categorical}}%
\global\long\def\GP{\textnormal{GP}}%
\global\long\def\logit{\operatorname{logit}}%
\global\long\def\Identity{\textnormal{I}}%
\global\long\def\indep{\stackrel{\textnormal{indep}}{\sim}}%
\global\long\def\KL{K}%

\begin{abstract}
We present a Bayesian nonparametric model for conditional distribution
estimation using Bayesian additive regression trees (BART). The generative
model we use is based on rejection sampling from a base model. Typical
of BART models, our model is flexible, has a default prior specification,
and is computationally convenient. To address the distinguished role
of the response in the BART model we propose, we further introduce
an approach to targeted smoothing which is possibly of independent
interest for BART models. We study the proposed model theoretically
and provide sufficient conditions for the posterior distribution to
concentrate at close to the minimax optimal rate adaptively over smoothness
classes in the high-dimensional regime in which many predictors are
irrelevant. To fit our model we propose a data augmentation algorithm
which allows for existing BART samplers to be extended with minimal
effort. We illustrate the performance of our methodology on simulated
data and use it to study the relationship between education and body
mass index using data from the medical expenditure panel survey (MEPS).
\end{abstract}
\doublespacing

\section{Introduction}

We consider here the Bayesian nonparametric estimation of a conditional
distribution of a response $Y_{i}$ based on predictors $X_{i}$.
A common strategy is to introduce a latent variable $b$, and set
$Y_{i}\sim h(y\mid X_{i},b,\theta)$ given $b$, where $h(y\mid x,b,\theta)$
is a parametric model. This includes mixture models where $b$ is
a latent class indicator and $f(y\mid x)=\sum_{k=1}^{\infty}\omega_{k}(x)\ h(y\mid x,\theta_{k})$
\citep{dunson2009nonparametric,rodriguez2011nonparametric,dunson2008kernel,maceachern1999dependent},
as well as Gaussian process latent variable/covariate models where
$b$ is continuous \citep{wang2012gaussian,kundu2014latent,dutordoir2018gaussian}.

A conceptually simpler approach models $f(y\mid x)$ by tilting a
base model:
\begin{align}
f(y & \mid x)=\frac{h(y\mid x,\theta)\ \Phi\{r(y,x)\}}{\int h(\yt\mid x,\theta)\ \Phi\{r(\yt,x)\}\ d\yt}.\label{eq:main-model}
\end{align}
We refer to $h(y\mid x,\theta)$ as the \emph{base model} and $\Phi(\mu)$
as the \emph{link function}. When $r(y,x)$ is a constant, this model
reduces to the base model, allowing the user to center the model on
a desired parametric model. A special case of (\ref{eq:main-model})
takes $\Phi(\mu)=e^{\mu}$ and $r(y,x)$ to be a Gaussian process
\citep{tokdar2010bayesian}. In the context of (marginal) density
estimation, \citet{murray2009gaussian} proposed the Gaussian process
density sampler (GP-DS), which sets $\Phi(\mu)$ to be a sigmoidal
function such as a logistic function $\Phi(\mu)=(1+e^{-\mu})^{-1}$.
Methods based on Gaussian processes have elegant theoretical properties
\citep{van2008rates} but are somewhat difficult to work with due
to the integral in the denominator of (\ref{eq:main-model}) and the
need to compute, store, and invert an $N\times N$ matrix. The goal
of this paper is to propose a method based on (\ref{eq:main-model})
with the following desirable properties.
\begin{itemize}
\item Algorithms for posterior inference are straight-forward to implement.
\item The posterior possesses strong theoretical properties, obtaining posterior
convergence rates close to the best possible.
\item For routine use, a default prior can be used which empirically obtains
good practical performance.
\item It is easy to shrink towards the base model $h(y\mid x,\theta)$ so
that the model naturally adapts to the complexity of the data.
\end{itemize}
We propose a modification of the Bayesian additive regression trees
(BART) model of \citet{chipman2010bart} which we refer to as the
SBART density sampler (SBART-DS). We choose $r(y,x)$ to be a \emph{soft}
decision tree \citep{linero2017abayesian,irsoy2012soft} which smooths
in a targeted fashion on the the response variable $y$ \citep{starling2018bart}.
A benefit of the BART framework is that we are able to develop default
priors based on well-known heuristics and show that these default
priors perform well in practice.

To construct inference algorithms, we restrict the choice of $\Phi(\mu)$
to the logit, probit, or $t_{\nu}$-link functions. Our proposal is
similar to the GP-DS, but is adapted to conditional distribution estimation.
We construct an efficient MCMC algorithm to sample from the posterior
distribution by combining a data augmentation scheme of \citet{rao2016data}
with an additional layer of data augmentation. After performing this
data augmentation, we can update the parameters of the model using
the same Bayesian backfitting algorithm as \citet{chipman2010bart}.
Given that one has the ability to perform Bayesian backfitting, the
algorithms we construct are simple to implement.

We present theoretical results which show that suitably-specified
SBART-DS priors attain convergence rates which are close to the best
possible. Simplifying slightly, we show that in the high-dimensional
sparse setting, where only $D-1\ll P$ of the $P$ predictors are
relevant, SBART-DS can obtain the oracle rate of convergence $\epsilon_{n}=n^{-2\alpha/(2\alpha+D)}$
up-to a logarithmic term where $\alpha$ is related to the smoothness
level of the true conditional density. In a simulation study we show
that these theoretical results are suggestive of what occurs in practice,
as SBART-DS is capable of filtering out irrelevant predictors. 

In Section \ref{sec:Model-Description} we review BART/SBART and describe
a naive version of SBART-DS; we then describe our approach for targeted
smoothing which centers the prior on $r(\cdot,x)$ on a desired Gaussian
process. In Section \ref{sec:Posterior-Computation} we provide data
augmentation algorithms for fitting (\ref{eq:main-model}) when the
link function $\Phi(\mu)$ is the probit, logit, or Student's $t_{\nu}$
link. In Section \ref{sec:Theoretical-Results} we present our theoretical
results. In Section \ref{sec:Simulation-Study} we conduct a simulation
study which shows that SBART-DS outperforms a method based on Dirichlet
process mixtures when the number of predictors is moderate. We then
apply SBART-DS to data from the Medical Expenditure Panel Survey (MEPS)
to study the relationship between educational attainment and body
mass index in adult women. We conclude in Section \ref{sec:Discussion}
with a discussion.

\section{Model Description \label{sec:Model-Description}}

\subsection{Review of Bayesian Additive Regression Trees}

The Bayesian additive regression trees (BART) framework models a function
$r(x)$ as a sum of regression trees $r(x)=\sum_{m=1}^{M}g(x;\Tree_{m,}\sM_{m})$
where $\Tree_{m}$ denotes the topology and splitting rules of a binary
decision tree and $\sM_{m}=\{\mu_{m1},\ldots,\mu_{mL_{m}}\}$ gives
a prediction for each of the $L_{m}$ terminal (leaf) nodes of $\Tree_{m}$.
Figure \ref{fig:treefig} gives a schematic which shows how predictions
are obtained from a given (single) tree. \citet{chipman2010bart}
specify a prior $\pi_{\Tree}$ for the tree topologies and a prior
$\pi_{\sM}$ on the $\mu_{m\ell}$'s given $\Tree_{m}$. We write
$r\sim\BART(\pi_{\Tree},\pi_{\sM})$ to denote that $r$ has the associated
BART prior. Typically, we set $\mu_{m\ell}\iid\Normal(0,\sigma_{\mu}^{2}/M)$
so that $\Var\{r(x)\}=\sigma_{\mu}^{2}$ regardless of the number
of trees used in the model.

\begin{figure}
\begin{centering}
\includegraphics[width=0.9\textwidth]{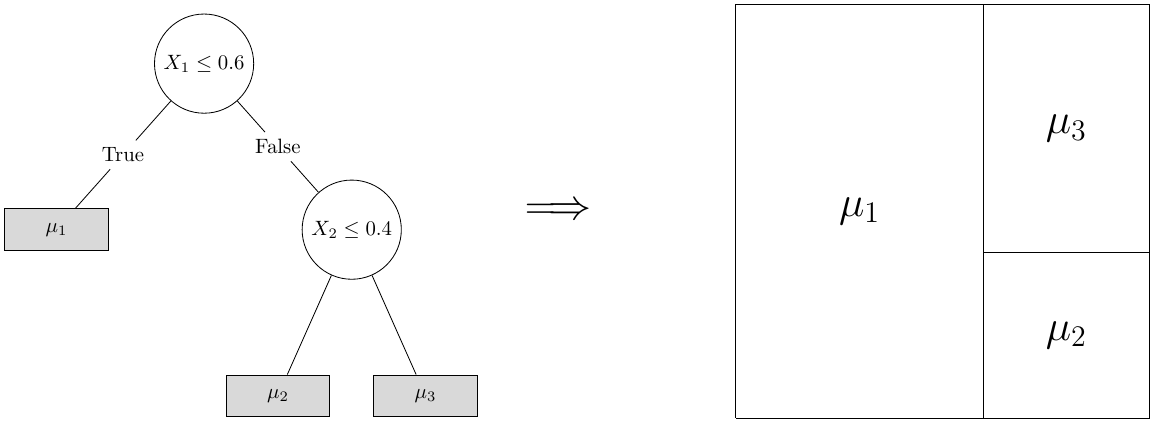}
\par\end{centering}
\caption{Schematic which shows how a decision tree induces a regression function.
Associated to the decision tree on the left is a partition of $[0,1]^{2}$
with the function $g(x;\protect\Tree,\protect\sM)$ returning $\mu_{1},\mu_{2}$,
or $\mu_{3}$.}

\label{fig:treefig}

\end{figure}

A problem with methods based on decision trees is that realizations
of $r\sim\BART(\pi_{\Tree},\pi_{\sM})$ will not be continuous in
$x$. This is particularly problematic for density estimation, as
we generally prefer estimates of the density to be continuous. A smooth
variant of BART called soft BART (SBART) was introduced by \citet{linero2017abayesian}.
This takes the tree $\Tree_{m}$ to be a \emph{smooth} decision tree,
where observations are assigned a weight $\varphi_{m\ell}(x)$ to
leaf node $\ell$ of tree $m$. As a point of comparison, non-soft
decision trees use the weights $\varphi_{m\ell}(x)=\prod_{b\in\Ancestors_{m\ell}}I(x_{j_{b}}\le C_{b})^{1-R_{b}}\ I(x_{j_{b}}>C_{b})^{R_{b}},$
where $\Ancestors_{m\ell}$ denotes the collection of branches which
are \emph{ancestors} of leaf $\ell$ of tree $m$, $j_{b}$ denotes
the coordinate along which $b$ splits, $C_{b}$ denotes the cutpoint
of branch $b$, and $R_{b}$ is the indicator that the path from the
root to the leaf goes right at $b$. A soft decision tree instead
takes 
\[
\varphi_{m\ell}(x)=\prod_{b\in\Ancestors_{m\ell}}\psi(x;C_{b},\tau_{b})^{1-R_{b}}\ \{1-\psi(x;C_{b},\tau_{b})\}^{R_{b}},
\]
where $\psi(x;c,\tau)$ is the cumulative distribution function of
a location-scale family with location $c$ and scale $\tau$. If $\psi(x)=I(x\le0)$
(or, equivalently, as $\tau\to0)$ we get a standard decision tree.
If we instead take $\psi(x)$ to be a smooth function then $r(y,x)$
will also be smooth. The parameter $\tau$ is analogous to a bandwidth
parameter, with larger values of $\tau$ giving smoother functions.
Like \citet{linero2017abayesian} we will take $\psi(x)=(1+e^{-x})^{-1}$
and use tree-specific bandwidths $\tau_{m}$. We write $r\sim\SBART(\pi_{\Tree},\pi_{\sM})$
to denote that $r$ has an SBART prior, where $\pi_{\Tree}$ is now
a prior over the \emph{soft} trees $\Tree_{m}$.

For completeness, we describe the prior over the tree structures we
will use. We assume that each coordinate $x_{j}$ of the predictors
has been scaled to lie in $[0,1${]}. This can be done, for example,
by applying the empirical quantile transform to a subset of the observed
values for each predictor and then interpolating the remaining values.
A tree $\Tree_{m}$ is sampled in the following steps. 
\begin{enumerate}
\item Initialize $\Tree_{m}$ with an single node of depth $D_{m}=0$.
\item For all nodes of depth $D_{m}$, make that node a branch node, with
a left and right child of depth $D_{m}+1$, with probability $\alpha(1+D_{m})^{-\beta}$;
otherwise, make the node a leaf node.
\item For all branch nodes $b$ of depth $D_{m}$, sample the splitting
coordinate $j_{b}\sim\Categorical(s)$ and a splitting point $C_{b}\sim\Uniform(L_{j_{b}},U_{j_{b}})$
where $\prod_{j=1}^{P}[L_{j},U_{j}]$ denotes the hyperrectangle of
$x$-values which lead to node $b$.
\item If all nodes of depth $D_{m}$ are leaf nodes, terminate; otherwise,
set $D_{m}\gets D_{m}+1$ and return to Step 2.
\end{enumerate}
The distribution of the splitting coordinate $j_{b}\sim\Categorical(s)$
determines how relevant a-priori we expect predictor to be; for example,
if $s_{1}=0.99$ we expect most splitting rules to use $x_{1}$, whereas
if $s_{1}=10^{-10}$ we expect none of the splitting rules to use
$x_{1}$. \citet{linero2016bayesian} took advantage of this fact
to perform automatic relevant determination \citep{neal1995bayesian}
for BART models by using a sparsity-inducing Dirichlet prior $s\sim\Dirichlet(a/P,\ldots,a/P)$
for some $a\ll P$. We also use this prior for the splitting proportion,
which will allow us to perform automatic relevance determination in
the density regression setting. This prior is crucial for proving
that the posterior adapts to the presence of irrelevant predictors.

\subsection{The Soft BART Density Sampler}

Our modeling strategy is based on the representation (\ref{eq:main-model})
$f(y\mid x)\propto h(y\mid x,\theta)\ \Phi\{r(y,x)\}$ where $\Phi(\mu)$
is a continuous, non-negative, monotonically increasing \emph{link}
function. Taking $r(y,x)=\Phi^{-1}\{f(y\mid x)/h(y\mid x,\theta)\}$
we see that (\ref{eq:main-model}) is valid whenever $h(y\mid x,\theta)$
and $f(y\mid x)$ have the same support for all $x$.

A naive approach is to set $r(y,x)\sim\BART(\pi_{\Tree},\pi_{\sM})$.
This specification has two problems. First, $r(y,x)$ will not be
smooth in $y$ so that draws from the prior and posterior of $f(y\mid x)$
will also not be smooth. The smoothness problem can be addressed by
setting $r(y,x)\sim\SBART(\pi_{\Tree},\pi_{\sM})$ instead. It is
this model that we study the theoretical properties of in Section
\ref{sec:Theoretical-Results}. 

Setting $r(y,x)\sim\SBART(\pi_{\Tree},\pi_{\sM})$ is still naive
because of the way in which BART shrinks $r(y,x)$ towards additive
models such that $r(y,x)=\sum_{v=1}^{V}r_{v}(y,x)$ where each $r_{v}(y,x)$
depends on a small subset of the coordinates of $(y,x)$ \citep{linero2017abayesian,rockova2017posterior}.
In the regression setting, we often expect that an underlying regression
function will have exactly this form; in the case of sparse additive
models \citep{ravikumar2007spam} for example, each $r_{v}(y,x)$
would depend on exactly one coordinate. This type of shrinkage-towards-additivity
is not appropriate for conditional density estimation due to the distinguished
nature of the response $y$; we instead want the predictors to \emph{interact}
with $y$.

To see why we want to force interactions with $y$, consider the strictly
additive function $r(y,x)=r_{Y}(y)+\sum_{p=1}^{P}r_{p}(x_{p})$. If
we take $\Phi(\mu)=e^{\mu},$ a massive cancellation occurs in (\ref{eq:main-model})
and the model reduces to $f(y\mid x)=h(y\mid x,\theta)\ \Phi\{r_{Y}(y)\}/\int h(\yt\mid x,\theta)\ \Phi\{r_{Y}(\yt)\}\ d\yt$,
effectively eliminating the predictors from the model. More generally,
any trees which do not split on $Y_{i}$ will have no effect on $f(y\mid x)$.
While exact cancellation is unique to the exponential link, it occurs
in an approximate form for the logistic link as well. At the other
extreme, SBART-DS uses a prior which favors utilizing a small number
of coordinates. If $y$ is eliminated, massive cancellation occurs
irrespective of the link function, and gives $f(y\mid x)=h(y\mid x,\theta)$.
Combined with a Dirichlet prior for $s$, this approach encodes prior
information that $f(y\mid x)$ is exactly equal to $h(y\mid x,\theta)$
with high probability.

\subsection{Targeted Smoothing via Random Basis Function Expansions}

We use the ``targeted smoothing'' approach of \citet{starling2018bart}
to overcome the problems of the naive SBART-DS prior. They set $r(y,x)=\gamma+\sum_{m=1}^{M}g(y;x,\Tree_{m},\sM_{m})$
where each leaf node is associated with a Gaussian process; that is,
for fixed $x$, we have $g(y;x;\Tree_{m},\sM_{m})\sim\GP\{0,\Sigma(\cdot,\cdot)\}$
where $\GP\{0,\Sigma(\cdot,\cdot)\}$ denotes a mean-$0$ Gaussian
process with covariance function $\Sigma(y,y')$.

\citet{starling2018bart} consider the case where the number of unique
values $y$ takes, $N_{y}$, is small. When $N_{y}$ is large this
is no longer practical due to the need to store and invert an $N_{y}\times N_{y}$
matrix for all $m$ trees. For SBART-DS we cannot guarantee that this
is the case. Instead, we set $r(y,x)=\gamma+\sum_{m=1}^{M}\Basis_{m}(y)\ g(x;\Tree_{m},\sM_{m})$
where each $\Basis_{m}$ is a random basis function.

To construct our approximation, consider the case where $\Tree_{m}$
us a non-soft decision tree. For fixed $x$ we can write $r(y,x)=\gamma+M^{-1/2}\sum_{m=1}^{M}\mu_{m}\ \Basis_{m}(y)$
where the $\mu_{m}$'s are iid $\Normal(0,\sigma_{\mu}^{2})$ random
variables. Under mild regularity conditions on the distribution of
the $\Basis_{m}$'s, as $M\to\infty$ a functional central limit theorem
will hold and this will converge weakly to a Gaussian process with
mean $\gamma$ and covariance function 
\begin{equation}
\Sigma(y,y')=\sigma_{\mu}^{2}\E\{\Basis_{1}(y)\ \Basis_{1}(y')\}.\label{eq:basis}
\end{equation}
This is the same as the distribution of $r(\cdot,x)$ used by \citet{starling2018bart}.
Rather than directly choose the basis functions $\Basis_{m}$, we
specify $\Sigma(y,y')$ and derive a distribution for $\Basis_{m}$
which matches (\ref{eq:basis}). We make use of the following proposition,
which follows from Bochner's Theorem, and is stated for completeness.
\begin{prop}
\label{prop:fourier} Let $\Sigma(y,y')=\sigma_{\mu}^{2}\delta(y-y')$
be a shift-invariant kernel with $\delta(0)=1$. Then there exists
a distribution $P(d\omega)$ such that $\Sigma(y,y')=\sigma_{\mu}^{2}\E\{2\cos(\omega y+b)\cos(\omega y'+b)\}$
where $\omega\sim P(d\omega)$ and $b\sim\Uniform(0,2\pi)$. Moreover,
$\delta(\cdot)$ is the characteristic function of $P(d\omega)$,
i.e., $\delta(t)=\int\exp\{i\omega t\}\ P(d\omega)$.
\end{prop}
The approach of using random Fourier features in this fashion was
introduced by \citet{rahimi2008random}. It follows from Proposition
\ref{prop:fourier} that we can take $\Basis_{m}(y)=\sqrt{2}\cos(\omega_{m}y+b_{m})$
where $\omega_{m}\iid P(d\omega)$ and $b_{m}\iid\Uniform(0,2\pi)$.
We list some possible choices below.
\begin{itemize}
\item $\omega_{m}\sim\Normal(0,\rho^{-2})$ corresponds to the squared exponential
covariance $\Sigma(y,y')=\sigma_{\mu}^{2}\exp\{-(y-y')^{2}/(2\rho^{2})\}$.
\item Setting $\omega_{m}\sim t_{\nu}$ with location $0$ and scale $\rho^{-1}$
gives the Matern kernel
\[
\Sigma(y,y')=\sigma_{\mu}^{2}\frac{1}{2^{v/2-1}\Gamma(v/2)}\left(\frac{\sqrt{v}|y-y'|}{\rho}\right)^{v/2}K_{v/2}\left(\frac{\sqrt{v}|y-y'|}{\rho}\right)
\]
where $K_{\nu}(\cdot)$ is a modified Bessel function of the second
kind. The exponential kernel $\Sigma(y,y')=\sigma_{\mu}^{2}\exp\{-|y-y'|/\rho\}$
is a special case ($v=1$).
\item In general, $P(d\omega)=p(\omega)\ d\omega$ can be obtained from
the inversion formula $p(\omega)=\frac{1}{2\pi}\int e^{-it\omega}\ \delta(t)\ dt$.
For example, inverting a Cauchy kernel $\delta(t)=\{1+t^{2}/\rho^{2}\}^{-1}$
shows that we can get a Cauchy kernel by sampling from the Laplace
distribution $p(\omega)=\frac{\rho}{2}e^{-\rho|\omega|}$.
\end{itemize}

\subsection{Shrinking Towards the Base Model}

A desirable feature of mixture models is the ability to center the
prior on a parametric submodel. Consider the infinite mixture $\sum_{k=1}^{\infty}\pi_{k}(x)\ h(y\mid x,\theta_{k})$.
If we choose the prior so that $\pi_{1}(x)\approx1$ with high probability,
then we are encoding prior knowledge that $h(y\mid x,\theta)$ is
itself highly likely to give an adequate representation of the data.
For models based on Dirichlet process mixtures $f(x,y)=\int h(x,y\mid\theta)\ F(d\theta)$
with $F\sim\Dirichlet(\alpha,F_{0}$), for example, this can be accomplished
by choosing a prior which shrinks $\alpha$ heavily towards $0$.
We can accomplish a similar goal with SBART-DS. Note that if $\Phi\{r(y,x)\}$
is a constant then it can be canceled in (\ref{eq:main-model}) so
that $f(y\mid x)$ reverts to $h(y\mid x,\theta)$. One approach is
to use a prior which encourages $\sigma_{\mu}^{2}$ to be close to
$0$, so that $\Phi\{r(y,x)\}\approx\Phi(\gamma)$, which is constant.
A second approach is to use a prior which encourages $\gamma$ to
be large and positive, so that $\Phi\{r(y,x)\}\approx1$.

To quantify these observations, let $\Delta_{x}=\sup_{y}|\sum_{m=1}^{M}\Basis_{m}(y)\ g(x;\Tree_{m},\sM_{m})/\sigma_{\mu}|$,
let $\Phi(\mu)$ satisfy Condition L in Section \ref{sec:Theoretical-Results},
and let $H(f,g)$ and $\KL(f,g)$ denote the Hellinger distance and
Kullback-Leibler divergence respectively. Note that the distribution
of $\Delta_{x}$ is free of $\sigma_{\mu}$. An application of Lemma
\ref{lem:compare} shows that $H(f,g)=O_{p}(\sigma_{\mu})$ and $\KL(f,g)=O_{p}(\sigma_{\mu}^{2})$.
It can further be shown that $\KL\{h(y\mid x,\theta),f(y\mid x)\}\le-\log\Phi(\gamma-\sigma_{\mu}\Delta_{x})$,
so that choosing a prior which makes $\gamma-\sigma_{\mu}\Delta_{x}$
large will also make the Kullback-Leibler divergence small; this can
be accomplished by centering $\gamma$ far away from $0$. More quantitative
results might be obtained from concentration inequalities for $\Delta_{x}$,
but we do not pursue this here.

\subsection{Default Prior Specification}

In our illustrations we use the following default prior specification.
Following \citet{chipman2010bart}, we fix the parameters $\alpha=0.95$
and $\beta=2$ in the prior for $\pi_{\Tree}$ and set $\mu_{m\ell}\sim\Normal(0,\sigma_{\mu}^{2}/M)$.
We fix $M=50$; in general, we recommend trying multiple values of
$M$. We set $\sigma_{\mu}\sim\text{Half-Cauchy}(0,1.5)$ to learn
an appropriate value of $\sigma_{\mu}$ from the data. By having mass
near $\sigma_{\mu}=0$, this also allows us to revert to the base
model $h(y\mid x,\theta)$. To induce further shrinkage to the base
model, we set $\gamma\sim\Normal(1,1)$. This has the additional benefit
of making the prior prefer models for which $\Phi\{r(y,x)\}$ is close
to $1$, which reduces the number of latent variables we need to introduce
when fitting the model by MCMC (see Section \ref{sec:Posterior-Computation}).
We use tree-specific bandwidths $\tau_{m}$ which are exponentially
distributed with mean $0.1$. To perform variable selection we specify
$s\sim\Dirichlet(a/P,\ldots,a/P)$ and use a hyperprior $a/(a+P)\sim\Beta(0.5,1)$. 

For targeted smoothing, we approximate the squared exponential kernel
by setting $\omega_{m}\sim\Normal(0,\rho^{-2})$. We set $\rho^{2}\sim\Gam(\alpha_{\rho},\beta_{\rho})$
to allow the length scale to be learned from the data. As a default,
we set $\alpha_{\rho}=1$ and $\beta_{\rho}=\pi^{2}/4$ after scaling
the $Y_{i}$'s to have unit variance. This choice is based on the
expected number of times a Gaussian process with length-scale $\rho$
is expected to cross $0$ on the interval $(-1,1)$: the expected
number of crossings is $2/(\pi\rho)$ so that if $\rho^{2}=4/\pi^{2}$
the expected number of crossings is $1$. Smaller values of $\rho^{2}$
correspond to more wiggly functions. Because our prior has positive
density at $0$, setting $\alpha_{\rho}=1$ allows for the possibility
that the function is very wiggly while defaulting to the prior belief
that it is not.

We choose the base model to be a Gaussian linear model with $h(y\mid x,\theta)=\varphi(y\mid\alpha_{\theta}+x^{\top}\beta_{\theta},\sigma_{\theta})$
where $\varphi(y\mid\mu,\sigma)$ is the density of a Gaussian random
variable with mean $\mu$ and variance $\sigma^{2}$. In our examples
we set $\pi(\alpha_{\theta},\beta_{\theta},\sigma_{\theta})\propto\sigma_{\theta}^{-1}$.
We can also shrink towards a semiparametric Gaussian model by setting
$h(y\mid x,\theta)=\varphi\{y\mid r_{\theta}(x),\sigma_{\theta}\}$
where $r_{\theta}(x)\sim\SBART(\pi_{\Tree}^{\theta},\pi_{\sM}^{\theta})$
and the default prior of \citet{linero2017abayesian} is specified
for $(\pi_{\Tree}^{\theta},\pi_{\sM}^{\theta})$.

\section{Posterior Computation \label{sec:Posterior-Computation}}

\subsection{Rejection Sampling Data Augmentation \label{subsec:Rejection-Sampling-Data}}

We use a two-layer data augmentation scheme which removes both the
intractable integral in the denominator of (\ref{eq:main-model})
and the link function $\Phi(\mu)$ from the likelihood. Our approach
is based on the following method for sampling from $f(y\mid x)$.
\begin{prop}
\label{prop:rejection} Suppose that we sample $Y_{1},Y_{2},Y_{3},\ldots\iid h(y\mid x,\theta)$
and sample $A_{j}\indep\Bernoulli[\Phi\{r(Y_{j},x)\}]$. Let $Z$
denote the $Y_{j}$ associated with the smallest index $J+1$ for
which $A_{J+1}=1$. Then conditional on $\{J,A_{j}:1\le j\le J+1\}$,
$Z$ is a draw from $f(y\mid x)\propto h(y\mid x,\theta)\Phi\{r(y,x)\}$
and $Y_{1},\ldots,Y_{J}$ are draws from $\bar{f}(y\mid x)\propto h(y\mid x,\theta)[1-\Phi\{r(y,x)\}]$.
\end{prop}
We make use of Proposition \ref{prop:rejection} by augmenting the
latent index $J$ and the sequence of rejected points. Associated
to each observation $Y_{i}=Y_{i0}$ we sample $Y_{ij}\sim h(y\mid X_{i},\theta)$
and $A_{ij}\sim\Bernoulli[\Phi(r(Y_{ij},X_{i})]$ until we reach the
first iteration $J_{i}+1$ such that $A_{i(J_{i}+1)}=1$. We then
work with the augmented state $\{Y_{ij}:1\le i\le N,0\le j\le J_{i}\}$,
which has likelihood 
\begin{equation}
\prod_{i=1}^{N}\prod_{j=0}^{J_{i}}h(Y_{ij}\mid X_{i},\theta)\times\prod_{i=1}^{N}\left(\Phi\{r(Y_{i0},X_{i})\}\prod_{j=1}^{J_{i}}[1-\Phi\{r(Y_{ij},X_{i})\}]\right).\label{eq:data-augment}
\end{equation}
For more details on the derivation of this expression, see \citet{rao2016data},
who consider the GP-DS model. At this stage \citet{rao2016data} propose
the use of Hamiltonian Monte Carlo to sample from the posterior distribution.
This is not an option for us, as the $\Tree_{m}$'s are discrete parameters.

\subsection{Bayesian Backfitting for Probit, Logit, and Student's $t_{\nu}$
Links}

We now apply data augmentation strategy of \citet{albert1993bayesian}.
Suppose that $\Phi(\mu)$ is cdf of either the probit, logit, or Student's
$t_{\nu}$ link. We can then associate to each $A_{ij}$ from Section
\ref{subsec:Rejection-Sampling-Data} a random variable 
\[
Z_{ij}=r(Y_{ij},X_{i})+\epsilon_{ij},\qquad\epsilon_{ij}\sim\Normal(0,\lambda_{ij}^{-1}),\qquad\lambda_{ij}\sim g(\lambda).
\]
Setting $A_{ij}=I(Z_{ij}\ge0)$ recovers the $\Bernoulli[\Phi\{r(Y_{ij},X_{i})\}]$
model. This model captures the three links we consider:
\begin{itemize}
\item For the probit link, $\lambda_{ij}^{-1}$ has a point-mass distribution
at $1$. 
\item When $\Phi(\mu)=T_{\nu}(\mu)$ is the Student's $t$ link with $\nu$
degrees of freedom, $\lambda_{ij}\sim\Gam(\nu/2,\nu/2)$. 
\item When $\Phi(\mu)=(1+e^{-\mu})^{-1}$ is the logistic link, $\lambda_{ij}^{-1/2}/2$
has a Kolmogorov-Smirnov distribution \citep{holmes2006bayesian}.
\end{itemize}
Compared to (\ref{eq:data-augment}), introducing the latent variables
$(Z_{ij},\lambda_{ij})$ leads to a more tractable likelihood:
\begin{equation}
\prod_{i=1}^{N}\prod_{j=0}^{J_{i}}h(Y_{ij}\mid X_{i},\theta)\times\Normal\{Z_{ij}\mid r(Y_{ij},X_{i}),\lambda_{ij}^{-1}\}\times g(\lambda_{ij}).\label{eq:final-augment}
\end{equation}
After reaching expression (\ref{eq:final-augment}) we can apply a
Bayesian backfitting algorithm to update $r(y,x)$. While the Bayesian
backfitting algorithm originally proposed by \citet{chipman2010bart}
does not account for heteroskedasticity in the $Z_{ij}$'s, several
recent works have shown how to accommodate this \citep{bleich2014bayesian,pratola2017heteroscedastic,linero2018shared}.
Consider the prior $\gamma\sim\Normal(\mu_{\gamma},\lambda_{\gamma}^{-1})$
and let $R_{ij}=Z_{ij}-\gamma-\sum_{m\ne k}\Basis_{m}(Y_{ij})\ g(X_{i};\Tree_{m},\sM_{m})$.
When updating $\Tree_{k}$ it suffices to consider the backfit model
\begin{equation}
R_{ij}=\Basis_{k}(Y_{ij})\ g(X_{i};\Tree_{k},\sM_{k})+\epsilon_{ij},\qquad\epsilon_{ij}\sim\Normal(0,\lambda_{ij}^{-1}),\label{eq:backfit-model}
\end{equation}
where recall that $\varphi_{k\ell}(X_{i})$ is the weight associated
to leaf $\ell$ of tree $k$ at $X_{i}$. Let $\varphi_{k}(X_{i})$
be a vector with $\ell^{\text{th}}$ entry $\varphi_{k\ell}(X_{i})$.
Then we can rewrite (\ref{eq:backfit-model}) as $R_{ij}=\Basis_{k}(Y_{ij})\ \varphi_{k}(X_{i})^{\top}\mu_{k}+\epsilon_{ij}$
or, in multivariate form, $\bm{R}\sim\Normal(\Basis_{k}\mu,\Lambda^{-1})$
where the rows of $\Basis_{k}$ correspond to $\Basis_{k}(Y_{ij})\varphi_{k}(X_{i})^{\top}$
and $\Lambda$ is diagonal with entries $\lambda_{ij}$. If $\mu_{k}\sim\Normal(0,\lambda_{\mu}^{-1})$
where $\lambda_{\mu}=M/\sigma_{\mu}^{2}$ then it follows from standard
properties of the multivariate Gaussian distribution that

\begin{equation}
\begin{aligned}\bm{R}\sim\Normal(0,\Lambda^{-1}+\Basis_{k}\Basis_{k}^{\top}/\lambda_{\mu})]\qquad & \text{and}\\{}
[\mu_{k}\mid\bm{R}]\sim\Normal(V\Basis_{k}^{\top}\Lambda\bm{R},V)\qquad & \text{where }V=(\Basis_{k}^{\top}\Lambda\Basis_{k}+\lambda_{\mu}\Identity)^{-1}.
\end{aligned}
\label{eq:marginal-fc}
\end{equation}
After applying the Woodbury matrix identity and the matrix determinant
lemma \citep[matrix identities]{brookes2011matrix}, the likelihood
of $\Tree_{k}$ after integrating out $\mu_{k}$ is given by 
\begin{equation}
(2\pi)^{-N/2}\prod_{i=1}^{N}\lambda_{i}^{1/2}\det(\Identity+\Basis_{k}^{\top}\Lambda\Basis_{k}/\lambda_{\mu})\ \exp\left[-\frac{1}{2}\left\{ \bm{R}^{\top}\Lambda\bm{R}-\delta^{\top}(\Identity+\Basis_{k}^{\top}\Lambda\Basis_{k}/\lambda_{\mu})^{-1}\delta\right\} \right]\label{eq:marg-lwoodbury}
\end{equation}
where $\delta=\Basis_{k}^{\top}\Lambda\bm{R}$. The value of (\ref{eq:marg-lwoodbury})
is that it avoids taking the determinant of and inverting the $N\times N$
matrix $\Lambda^{-1}+\Basis_{k}\Basis_{k}^{\top}/\lambda_{\mu}$.
The marginal likelihood $L_{k}(\Tree,\Basis)$ given by (\ref{eq:marg-lwoodbury})
is used to update both the tree topology $\Tree_{k}$ and the random
basis function $\Basis_{k}(y)$ using Metropolis-Hastings.

Our final MCMC scheme is summarized in Algorithm \ref{alg:DA}, which
calls Algorithm \ref{alg:Metropolis} to update $(\Tree_{k},\Basis_{k},\sM_{k})$.
The Markov transition function $Q(\Tree_{k}\to\Tree')$ used to propose
new tree topologies is a mixture of the \texttt{BIRTH}, \texttt{DEATH},
and \texttt{CHANGE} proposals described by \citet{chipman1998bayesian}
and a $\texttt{PRIOR}$ proposal which samples $\Tree'$ from the
prior.

\begin{algorithm}
\caption{An iteration of the data augmentation algorithm for SBART-DS}

\begin{enumerate}
\item For $i=1,\ldots,N$, set $Y_{i0}=Y_{i}$ and sample $Y_{i1},Y_{i2},\ldots\sim h(y\mid X_{i},\theta)$
and $A_{i1},A_{i2},\ldots\sim\Bernoulli[\Phi\{r(Y_{ij},X_{i})\}]$
until $A_{i(J_{i}+1)}=1$. Retain the samples $Y_{i0},\ldots,Y_{iJ_{i}}$.
\item Make an update to $\theta$ which leaves the full conditional $\pi(\theta\mid-)\propto\pi(\theta)\prod_{i,j}h(Y_{ij}\mid X_{i},\theta)$
invariant. 
\item Sample $Z_{ij}\sim f(z\mid\mu_{ij})$ truncated to $(0,\infty)$ for
$j=0$ and $(-\infty,0)$ for $j>0$ where $\mu_{ij}=r(Y_{ij},X_{i})$
and $f(z\mid\mu_{ij})$ is a normal, logistic, or Student's $t_{\nu}$
distribution with location $\mu_{ij}$ and scale $1$ for the probit,
logit, and $T_{\nu}$ links respectively. 
\item Sample $\lambda_{ij}$ from its full conditional given $Z_{ij}$ for
all $1\le i\le N$ and $0\le j\le J_{i}$.
\begin{itemize}
\item For the probit link, $\lambda_{ij}\equiv1$. 
\item For the Student's $t_{\nu}$ link, $\lambda_{ij}\sim\text{Gam}\{(\nu+1)/2,(\nu+[Z_{ij}-r(Y_{ij},X_{i})]^{2})/2\}.$
\item For the logit link, sample $\lambda_{ij}^{-1}$ using the rejection
sampling algorithm of \citet{holmes2006bayesian}.
\end{itemize}
\item For $m=1,\ldots,M$ update $(\Tree_{m},\sM_{m},\Basis_{m})$ using
the Metropolis-Hastings algorithm given in Algorithm \ref{alg:Metropolis}.
\end{enumerate}
\label{alg:DA}
\end{algorithm}

\begin{algorithm}

\caption{Metropolis-Hastings update for $(\protect\Tree_{k},\protect\sM_{k},\protect\Basis_{k})$}

\begin{enumerate}
\item Compute $\bm{R}$ as in (\ref{eq:backfit-model}). 
\item Propose a tree $\Tree'$ from a Markov transition kernel $Q(\Tree_{k}\to\Tree')$. 
\item Set $\Tree_{k}=\Tree'$ with probability 
\[
\min\left\{ \frac{\pi_{\Tree}(\Tree')\ L_{k}(\Tree',\Basis_{k})\ Q(\Tree'\to\Tree_{k})}{\pi_{\Tree}(\Tree_{k})\ L_{k}(\Tree_{k},\Basis_{k})\ Q(\Tree_{k}\to\Tree')},1\right\} .
\]
Otherwise, do not change $\Tree_{k}$. 
\item Sample a basis function $\Basis'(y)=\sqrt{2}\cos(\omega'y+b')$ by
sampling $\omega'\sim P(d\omega)$ and $b'\sim\Uniform(0,2\pi)$.
Then set $\Basis_{k}=\Basis'$ with probability 
\[
\min\left\{ \frac{L_{k}(\Tree_{k},\Basis')}{L_{k}(\Tree_{k},\Basis_{k})},1\right\} .
\]
Otherwise, do not change $\Basis_{k}$. 
\item Sample $\mu_{k}\sim\Normal(V\Basis_{k}^{\top}\Lambda\bm{R},V)$ where
$V=(\Basis_{k}^{\top}\Lambda\Basis_{k}+\Identity/\lambda_{\mu})^{-1}$
and $\lambda_{\mu}=M/\sigma_{\mu}^{2}$.
\end{enumerate}
\label{alg:Metropolis}

\end{algorithm}

\section{Theoretical Results \label{sec:Theoretical-Results}}

We show that SBART-DS attains close to the minimax-optimal concentration
rate for $(P+1)$-dimensional functions $r(y,x)$ in the high-dimensional
sparse setting. All proofs are deferred to the appendix. We consider
the case where $r(y,x)$ depends on only $D$ coordinates of $(y,x)^{\top}$
where the relevant subset is unknown and must be learned from the
data. Following \citet{pati2013posterior} we study concentration
with respect to the integrated Hellinger distance. Let $H(f,f_{0})=\{\int(\sqrt{f_{0}(y\mid x})-\sqrt{f(y\mid x})^{2}\ dy\ F_{X}(dx)\}^{1/2}$
denote the \emph{$F_{X}$-integrated Hellinger distance} between $f_{0}(y\mid x)$
and $f(y\mid x)$. The covariates $X_{i}$ are assumed to be iid from
$F_{X}$, which is not assumed to be known. We similarly define an
\emph{$F_{X}$-integrated Kullback-Leibler} neighborhood. Define $K(f_{0},f)=\int f_{0}\log\frac{f_{0}}{f}\ dy\ dF_{X}$
and $V(f_{0},f)=\int f_{0}\left(\log\frac{f_{0}}{f}\right)^{2}\ dy\ dF_{X}$.
Then the integrated Kullback-Leibler neighborhood is given by $K(\epsilon)=\left\{ f:K(f_{0},f)\le\epsilon^{2}\ \text{and }V(f_{0},f)\le\epsilon^{2}\right\} .$
Let $\Data_{n}$ denote the data $\{X_{i},Y_{i}:i=1,\ldots,n\}$ and
let $\Pi$ denote a prior distribution on $r$ and additional hyperparameters.
We say that the posterior has a convergence rate of at least $\epsilon_{n}$
if there exists a constant $C>0$ such that $\Pi\{H(f_{r},f_{0})\ge C\epsilon_{n}\mid\Data_{n}\}\to0$
in probability. To simplify the theoretical results, we assume that
$X_{i}$ and $Y_{i}$ take values in $[0,1]^{P+1}$. We additionally
make the following assumptions about the true data generating process
$F_{0}$.

\paragraph{Condition F (on $F_{0}$):}

The true conditional density $f_{0}(y\mid x)$ can be written as $f_{r_{0}}(y\mid x)$
for some $r_{0}\in C^{\alpha,R}([0,1]^{P+1})$ where $C^{\alpha,R}([0,1])$
is the ball of radius $R$ in the space of $\alpha$-H\"older smooth
functions on $[0,1]^{P+1}$, where $f_{r}(y\mid x)$ is defined as
\[
f_{r}(y\mid x)=\frac{h(y)\ \Phi\{r(y,x)\}}{\int h(\widetilde{y})\ \Phi\{r(\widetilde{y},x)\}\ d\widetilde{y}}
\]
for some density $h(y)$ on $[0,1]$. Additionally, we can write $r_{0}(y,x)=\widetilde{r}(y,x_{\mathcal{S}})$
where $x_{\mathcal{S}}=\{x_{j}:j\in\mathcal{S}\}$ and $\mathcal{S}$
is a subset of $\{1,\ldots,P\}$ of cardinality $D-1$. That is, $r_{0}(y,x)$
depends on at most $D$ coordinates of $(y,x)^{\top}$. The number
of predictors $P\equiv P_{n}$ depends on $n$ but is such that $\log(P+1)\le C_{\eta}n^{\eta}$
for some $\eta\in(0,1)$.
\begin{rem}
For simplicity, we consider $\widetilde{r}$ and $\mathcal{S}$ to
be independent of $n$; in particular, we do not consider $D$ diverging
with $n$. There exists \emph{some} $r_{0}$ such that $f_{0}=f_{r_{0}}$
by taking $r_{0}(y\mid x)=\Phi^{-1}\{f_{0}(y\mid x)/h(y)\}$, provided
that $h(y)$ and $f_{0}(y\mid x)$ have common support. When $h(y)=1$,
the assumption that $r_{0}$ is continuous on $[0,1]^{P+1}$ implies
that $C^{-1}\le f_{0}(y\mid x)\le C$ for some constant $C$, i.e.,
$f_{0}(y\mid x)$ is bounded and bounded away from $0$.
\end{rem}

\paragraph{Condition L (on $\Phi)$:}

The link function $\Phi(\mu)$ is strictly increasing and is the cumulative
distribution function of a random variable $Z$ which is symmetric
about $0$ and has density $\phi(\mu)$ satisfying $\phi(\mu)/\Phi(\mu)\le\mathcal{K}$
for all $\mu$ and some constant $\mathcal{K}$.
\begin{rem}
We show in the appendix that Condition L holds for the logit ($\mathcal{K}=1$)
and $t_{\nu}$ ($\mathcal{K}=\sqrt{\nu}$) links, but fails for the
probit link. 
\end{rem}

\paragraph{Condition P (on $\Pi$):}

The function $r$ is given an $\text{SBART}(\pi_{\Tree},\pi_{\sM})$
prior with $M$ trees, conditional on $(\pi_{\Tree},\pi_{\sM},M)$.
Additionally, the prior $\Pi$ satisfies the following conditions.
\begin{enumerate}
\item [(P1)] There exists positive constants $(C_{M1},C_{M2})$ such that
the prior on the number of trees $M$ in the ensemble is $\Pi(M=t)=C_{M1}\exp\{-C_{M2}t\log t\}$.
\item [(P2)] A single bandwidth $\tau_{m}\equiv\tau$ is used and its prior
satisfies $\Pi(\tau\ge x)\le C_{\tau1}\exp(-x^{C_{\tau2}})$ and $\Pi(\tau^{-1}\ge x)\le C_{\tau3}\exp(-x^{C_{\tau4}})$
for some positive constants $C_{\tau1},\ldots,C_{\tau4}$ for all
sufficiently large $x$, with $C_{\tau2},C_{\tau4}<1$. Moreover,
the density of $\tau^{-1}$ satisfies $\pi_{\tau^{-1}}(x)\ge C_{\tau5}e^{-C_{\tau6}x}$
for large enough $x$ and some positive constants $C_{\tau5}$ and
$C_{\tau6}$.
\item [(P3)] The prior on the splitting proportions is $s\sim\text{Dirichlet}(a/P^{\xi},\ldots,a/P^{\xi})$
for some $\xi>1$ and $a>0$. 
\item [(P4)] The $\mu_{m\ell}$'s are iid from a density $\pi_{\mu}(\mu)$
such that $\pi_{\mu}(\mu)\ge C_{\mu1}e^{-C_{\mu2}|\mu|}$ for some
coefficients $C_{\mu1},C_{\mu2}$. Additionally, there exists constants
$C_{\mu3},C_{\mu4}$ such that $\Pi(|\mu_{m\ell}|\ge t)\le C_{\mu3}\exp\{-t^{C_{\mu4}}\}$
for all $t$. 
\item [(P5)] Let $D_{m}$ denote the depth of tree $\Tree_{m}$. Then $\Pi(D_{m}=k)>0$
for all $k=0,1,\ldots,2D$ and $\Pi(D_{m}>d_{0})=0$ for some $d_{0}\ge D$.
\item [(P6)] The gating function $\psi:\Reals\to[0,1]$ of the SBART prior
is such that $\sup_{x}|\psi'(x)|<\infty$ and the function $\rho(x)=\psi(x)\{1-\psi(x)\}$
is such that $\int\rho(x)\ dx>0$, $\int|x|^{m}\rho(x)\ dx<\infty$
for all integers $m\ge0$, and $\rho(x)$ can be analytically extended
to some strip $\{z:|\Im(z)|\le U\}$ in the complex plane.
\end{enumerate}
\begin{rem}
Conditions other than Condition P might also be used. Recent work
of \citet{rockova2017posterior}, for example, studies concentration
results for BART using different sets of conditions, and the conditions
overall are weaker than the conditions presented here. A downside
of these results is they apply only when non-smooth decision trees
are used, which induces non-smooth densities. Condition P2 holds when
$\tau$ is given an inverse-gamma prior truncated from above, while
Condition P4 holds when the $\mu_{m\ell}$'s are given a Laplace prior,
although as noted by \citet{linero2017abayesian} this could potentially
be weakened to allow a Gaussian prior with a hyperprior on $\sigma_{\mu}$
(we do not pursue this here). Condition P6 holds for the logistic
gating function $\psi(x)=\{1+\exp(-x)\}^{-1}$, which is used by default.
Condition P5 holds if we truncate the prior of \citet{chipman2010bart}
at some large $d_{0}$, which is extremely unlikely to affect the
MCMC in practice. Hence, satisfying P5 is not a practical concern.
Condition P1 is problematic because BART implementations do not use
a prior on $M$. In practice, we find that selecting $M$ by cross
validation is more reliable than using a prior; we recommend either
(i) using cross validation to select $M$ or (ii) fixing $M$ at a
default value such as $M=200$ (recommended by \citealp{chipman2010bart})
or $M=50$ (used here).
\end{rem}
\begin{thm}
\label{thm:sparse}Suppose that Condition L, Condition F, and Condition
P hold. Then there exists a positive constant $C$ such that $\Pi\{H(f_{0},f_{r})\ge C\epsilon_{n}\mid\mathcal{D}_{n}\}\to0$
in probability, where $\epsilon_{n}=n^{-\alpha/(2\alpha+D)}(\log n)^{t}+\sqrt{\frac{D\log(P+1)}{n}}$
and $t=\alpha(D+1)/(2\alpha+D)$.
\end{thm}
We prove Theorem \ref{thm:sparse} by checking (a)---(c) in Proposition
\ref{prop:sc}, which are analogous to conditions of \citet{ghosal2000convergence}.
\begin{prop}
\label{prop:sc} Let $\Pi$ denote a prior for a conditional density
$f(y\mid x)$ and let $\epsilon_{n}$ and $\bar{\epsilon}_{n}$ be
sequences of positive numbers such that $\bar{\epsilon}_{n},\epsilon_{n}\to0$,
$n\epsilon_{n}^{2}\to\infty$, and $\epsilon_{n}\le\bar{\epsilon}_{n}$.
Let $N(\epsilon,\mathcal{F},H)$ denote the $\epsilon$-covering number
of $\mathcal{F}$ with respect to $H$ (i.e., the number of balls
of radius $\epsilon$ required to cover $\mathcal{F}$). Suppose that
there exists positive constants $C,C_{N}$ such that for all sufficiently
large $n$ there exist sets of conditional densities $\Sieve_{n}$
satisfying the following conditions:
\begin{enumerate}
\item [(a)] Entropy Bound:\emph{ }$\log N(\bar{\epsilon}_{n},\Sieve_{n},H)\le C_{N}n\bar{\epsilon}_{n}^{2}$.
\item [(b)] Support Condition: $\Pi(\Sieve_{n}^{c})\le\exp\{-(C+4)n\epsilon_{n}^{2}\}$.
\item [(c)] Prior Thickness:\emph{ }$\Pi\{f\in K(\epsilon_{n})\}\ge\exp(-Cn\epsilon_{n}^{2})$.
\end{enumerate}
Then $\Pi\{H(f_{0},f)\ge A\bar{\epsilon}_{n}\mid\Data_{n}\}\to0$
in probability for some constant $A>0$.
\end{prop}
The proof that our SBART prior satisfies these conditions is similar
to the proof of Theorem 3.1 of \citet{van2008rates}, who established
posterior convergence rates for density estimation using logistic
Gaussian processes. We use a collection of results of \citet{linero2017abayesian},
who established results similar to (a)---(c) for a regression function
$r\sim\text{SBART}(\pi_{\Tree},\pi_{\sM})$ with respect to the supremum
norm $\|r-r_{0}\|_{\infty}=\sup_{x,y}|r(y,x)-r_{0}(y,x)|$. We then
use the following lemma, which links the supremum-norm neighborhoods
of $r_{0}$ with the integrated Hellinger and Kullback-Leibler neighborhoods
of $f_{0}$; this allows us to convert results about the $\|\cdot\|_{\infty}$-norm
neighborhoods to results about the integrated neighborhoods. This
lemma is similar to Lemma 3.1 of \citet{van2008rates}, but with exponential
link $\Phi(\mu)=e^{\mu}$ replaced with a link satisfying Condition
L.
\begin{lem}
\label{lem:compare} Let $\Phi(\mu)$ be a link function satisfying
Condition L. Then for any measurable functions $u,v:[0,1]^{P+1}\to\mathbb{R}$
we have the following:
\begin{itemize}
\item $H^{2}(f_{u},f_{v})\le\mathcal{K}^{2}\|u-v\|_{\infty}^{2}\exp(\mathcal{K}\|u-v\|_{\infty})$;
\item $K(f_{u},f_{v})\lesssim\|u-v\|_{\infty}^{2}\exp(\mathcal{K}\|u-v\|_{\infty})(1+2\mathcal{K}\|u-v\|_{\infty})$;
and
\item $V(f_{u},f_{v})\lesssim\|u-v\|_{\infty}^{2}\exp(\mathcal{K}\|u-v\|_{\infty})(1+2\mathcal{K}\|u-v\|_{\infty})^{2}$. 
\end{itemize}
The expression $a\lesssim b$ here denotes that $a\le Cb$ for some
constant $C$ depending only on $\mathcal{K}$. 
\end{lem}
More generally, one expects that Theorem \ref{thm:sparse} can be
improved to allow for additive decompositions $r_{0}(y,x)=\sum_{j=1}^{J}r_{0j}(y,x)$
where the $r_{0j}$'s are functions which are $D_{j}$-sparse and
$\alpha_{j}$-Hölder continuous. Results in this framework \citep{linero2017abayesian,rockova2017posterior,yang2015minimax}
suggest that we should be able to obtain a rate $\epsilon_{n}=\sum_{j=1}^{J}n^{-\alpha_{j}/(2\alpha_{j}+D_{j})}\log(n)^{t_{j}}+\sqrt{n^{-1}D_{j}\log(P+1)}$,
which is a substantial improvement on Theorem \ref{thm:sparse}. One
difficulty with extending these results is that Condition P2 only
allows a single bandwidth, while different $\tau$'s will be optimal
for different $\alpha_{j}$'s. Unlike the nonparametric regression
setting, however, it is unclear how one would interpret the additivity
assumption for SBART-DS. We leave examining the additive framework
to future work.

\section{Illustrations \label{sec:Simulation-Study}}

\subsection{Simulation Study}

We now assess the performance of the SBART-DS using the simulation
example described by \citet{dunsonetal2007}. The response $Y_{i}$
is sampled from a mixture model 
\[
Y_{i}\sim e^{-2x}\text{Normal}(x,0.1^{2})+(1-e^{-2x})\text{Normal}(x^{4},0.2^{2})\qquad\text{given }X_{i1}=x.
\]
We set $N=500$ and have $P-1$ additional predictors which do not
influence the response. The marginal density of the $X_{i}$'s is
uniform on $[0,1]^{P}$. For SBART-DS we use the default prior with
$M\equiv50$ and the probit link. We do not make any attempt to tune
the hyperparameters $(a,\sigma_{\mu},\rho,\alpha,\beta,\gamma)$ beyond
this. We take the base model to be a normal linear regression model
$h(y\mid x,\theta)=\Normal(y\mid\alpha_{\theta}+\beta_{\theta}^{\top}x,\sigma_{\theta}^{2})$.
We consider moderate dimensions $P$ for illustrative purposes, but
in higher dimensions one might wish to induce sparsity $\beta_{\theta}$.

\begin{figure}
\centering{}\includegraphics[width=.8\textwidth]{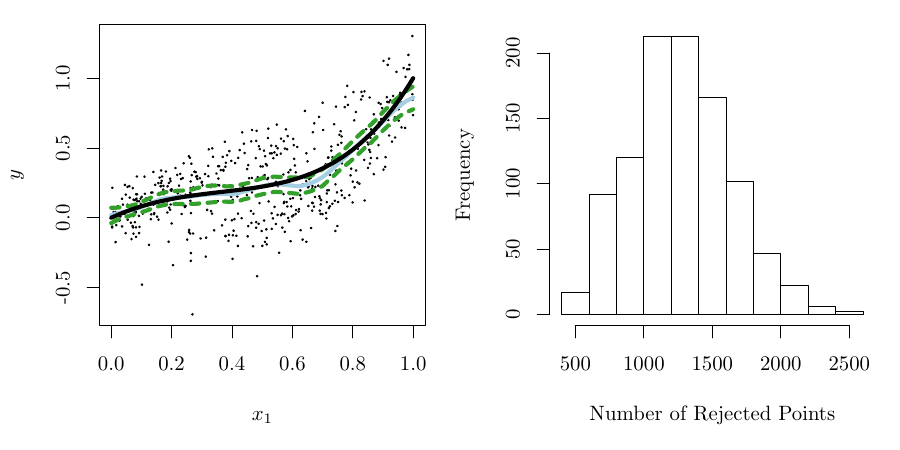}\caption{(Left) Plot of realized values of $X_{i1}$ against $Y_{i}$ for a
single replication of the experiment, with solid black line indicating
the true mean, light blue line indicating the estimated posterior
mean, and the dashed green lines indicating 95\% credible bands for
the mean function. (Right) The posterior distribution of the number
of rejected points estimated via Markov chain Monte Carlo.}
\label{fig:sim-setup}
\end{figure}

Figure \ref{fig:sim-setup} displays a scatterplot of the relationship
between $Y_{i}$ and $X_{i1}$ as well as the posterior mean and credible
band for the function $r(x)=\E(Y_{i}\mid X_{i}=x)$ with $P=5$. We
compare SBART-DS to a Dirichlet process mixture model described by
\citet{jaraetal2011} as implemented in the function \texttt{DPcdensity}
in the \texttt{DPpackage} package in \texttt{R}; we use this as a
comparison because there is publicly available software implementing
this methodology and \citet{jaraetal2011} show that it performs similarly
to the approach of \citet{dunsonetal2007}. This model uses the joint
specification
\begin{align*}
(X_{i},Y_{i}) & \sim\int\text{Normal}\left\{ \begin{pmatrix}x\\
y
\end{pmatrix}\mid\begin{pmatrix}\mu_{x}\\
\mu_{y}
\end{pmatrix},\begin{pmatrix}\Sigma_{xx} & \Sigma_{xy}\\
\Sigma_{yx} & \Sigma_{yy}
\end{pmatrix}\right\} \ dG(\mu,\Sigma),
\end{align*}
where $G\sim\text{DP}(\alpha G_{0})$ is a Dirichlet process with
a normal-inverse-Wishart base measure $G_{0}\equiv\text{Normal}(\mu\mid m,\kappa_{0}\Sigma)\ \text{IW}(\Sigma\mid\nu,\Psi)$.
The conditional density of $[Y_{i}\mid X_{i}=x]$ can be estimated
from an infinite mixture model as 
\[
f(y\mid x)=\sum_{k=1}^{\infty}\omega_{k}(x)\ \text{Normal}(y\mid\mu_{y\mid x},\Sigma_{y\mid x})
\]
where $\omega_{k}(x)\propto\pi_{k}\Normal(x\mid\mu_{x},\Sigma_{xx})$,
$\mu_{y\mid x}=\mu_{y}+\Sigma_{yx}\Sigma_{xx}^{-1}(x-\mu_{x})$, and
$\Sigma_{y\mid x}=\Sigma_{yy}-\Sigma_{yx}\Sigma_{xx}^{-1}\Sigma_{xy}$.
We use the same prior specification as \citet{jaraetal2011} but with
a larger value of $\nu$ to accommodate the fact that $\nu>P-1$ is
required.

\begin{figure}
\centering{}\includegraphics[width=0.8\textwidth]{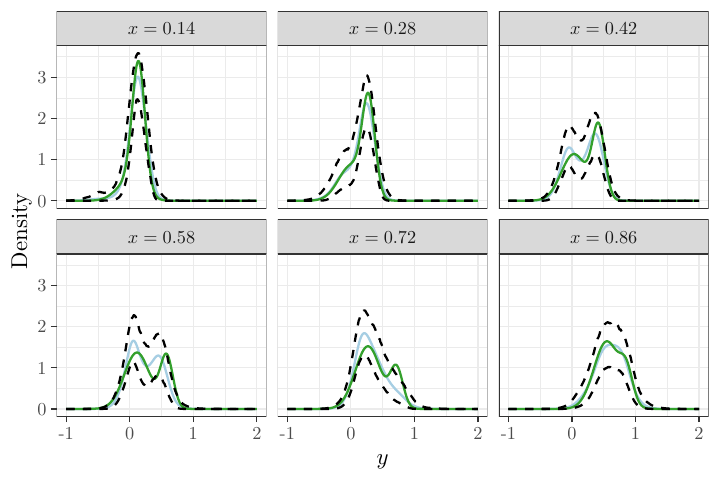}\caption{Posterior mean (blue), 95\% credible bands for the density (dashed
black) and true density function (green) for the simulated data, for
the values $X_{i1}\in\{0.14,0.28,0.42,0.58,0.72,0.86\}$. }
\label{fig:sim-density}
\end{figure}

Figure \ref{fig:sim-density} shows the fitted density for several
fixed values of $X_{i1}$ with all other predictors frozen at the
value $X_{ij}=0.5$ (as these predictors were correctly filtered out
of the model, their particular value is irrelevant). We see SBART-DS
successfully captures variability in the location, shape, and scale
of the densities, and produces 95\% credible bands which accurately
account for uncertainty in the estimates. SBART-DS also captures the
mean response accurately (left panel of Figure \ref{fig:sim-setup}).
Additionally, the number of rejected points is not prohibitively large,
and fitting SBART-DS was faster than fitting the Dirichlet process
mixture model using \texttt{DPcdensity}.

Table \ref{tab:results} compares SBART-DS to the Dirichlet process
mixture over $100$ replications of the above experiment with $P=20$.
We compare methods using the integrated total variation distance 
\[
\text{TV}(f_{0,}\widehat{f})=\int_{[0,1]^{P}}\int_{-\infty}^{\infty}|f_{0}(y\mid x)-\widehat{f}(y\mid x)|\ dy\ dx.
\]
This integral can be approximated via Monte Carlo integration by averaging
over a large out-of-sample test set of $X_{i}$'s and computing the
$dy$ integral numerically. We see that the Dirichlet process mixture
performs substantially worse than SBART-DS as measured by total variation
distance from the true data generating mechanism. It is somewhat surprising
that SBART-DS outperforms a Dirichlet process mixture for this example,
as the true model is a mixture model with a structure that one would
expect a Dirichlet process mixture to be primed to detect. The reason
that SBART-DS performs better is that the BART prior we used performs
variable selection and is capable of eliminating the $19$ irrelevant
predictors, whereas the Dirichlet process mixture is not designed
to detect sparsity. We expect that any method which does not explicitly
try to detect sparsity, such as the probit stick-breaking prior with
Gaussian processes proposed by \citet{rodriguez2011nonparametric}
or the kernel-based approach of \citet{dunsonetal2007} would also
be outperformed by SBART-DS, although we were unable to assess this
due to a lack of publicly available software for these approaches.

\begin{table}
\begin{centering}
\begin{tabular}{lrrr}
\toprule 
Method & Normalized Average $\text{TV}(f,\widehat{f})$ & $25^{\text{th}}$ percentile & $75$ percentile\tabularnewline
\midrule
SBART-DS & 1.00 & 0.94 & 1.05\tabularnewline
Dirichlet Process Mixture & 1.74 & 1.71 & 1.78\tabularnewline
\bottomrule
\end{tabular}
\par\end{centering}
\caption{Average of the integrated total variation distance $\text{TV}(f_{0},\widehat{f})$
over $100$ replications of the simulation study; to give a sense
of stability, we also give the $25^{\text{th}}$ and $75^{\text{th}}$
quantiles of these quantities over the 100 replications. For interpretability,
we normalized both scores by the average integrated total variation
distance of the SBART-DS model.}

\label{tab:results}
\end{table}

Summarizing our simulation study, we find that the ability of SBART-DS
to perform variable selection allows it to outperform Dirichlet process
mixtures. While we only conducted a formal simulation study for the
$P=20$ case, we found similar behavior for small values of $P$ as
well. Unlike other Bayesian nonparametric density regression approaches,
the ability to perform variable selection is automatic for SBART-DS.

\subsection{Analysis of MEPS Data \label{sec:xxx-Real-Data}}

We apply SBART-DS to data from the Medical Expenditure Panel Survey
(MEPS) from the year 2015. MEPS is an ongoing survey in the United
States which collects data on families/individuals, their medical
providers, and employers, with a focus on the cost and use of health
care.

There is a large literature which has considered the relationship
between socioeconomic status, education, and obesity. Educational
attainment relates to obesity in a complex fashion, with the effect
modified by the overall income of a region, gender, and other factors
\citep{cohen2013educational}. We examined this relationship on a
subset of the MEPS dataset consisting of responses from $1452$ women
aged between $25$ and $35$ years old, controlling for log-income
(measured as a percentage of the poverty line), age, and race. Existing
research predicts that higher educational attainment will be associated
with lower obesity levels in this group.

\begin{figure}
\begin{centering}
\includegraphics[width=0.8\textwidth]{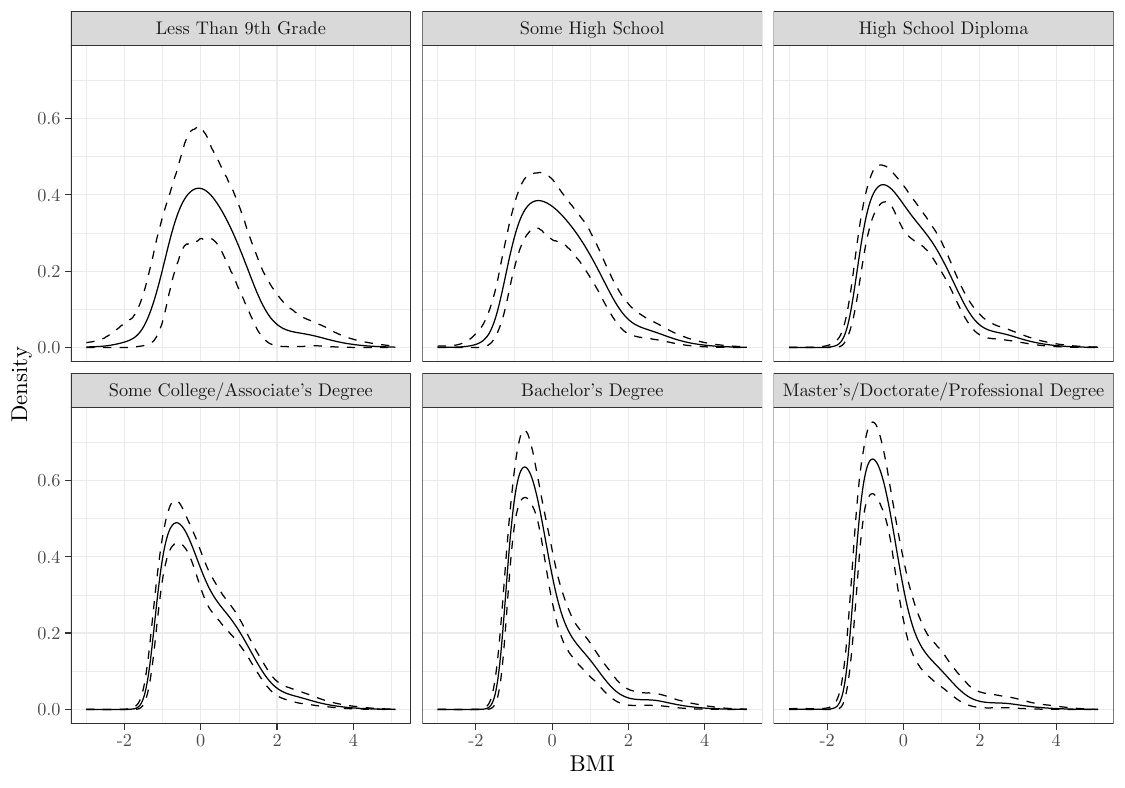}
\par\end{centering}
\centering{}\caption{Density estimates and 95\% credible bands for $f(y\mid x)$ for different
educational levels for white women aged between 25 and 35, fixing
log-income and age at their median values.}
\label{fig:density-est}
\end{figure}

\begin{figure}
\begin{centering}
\includegraphics[width=0.9\textwidth]{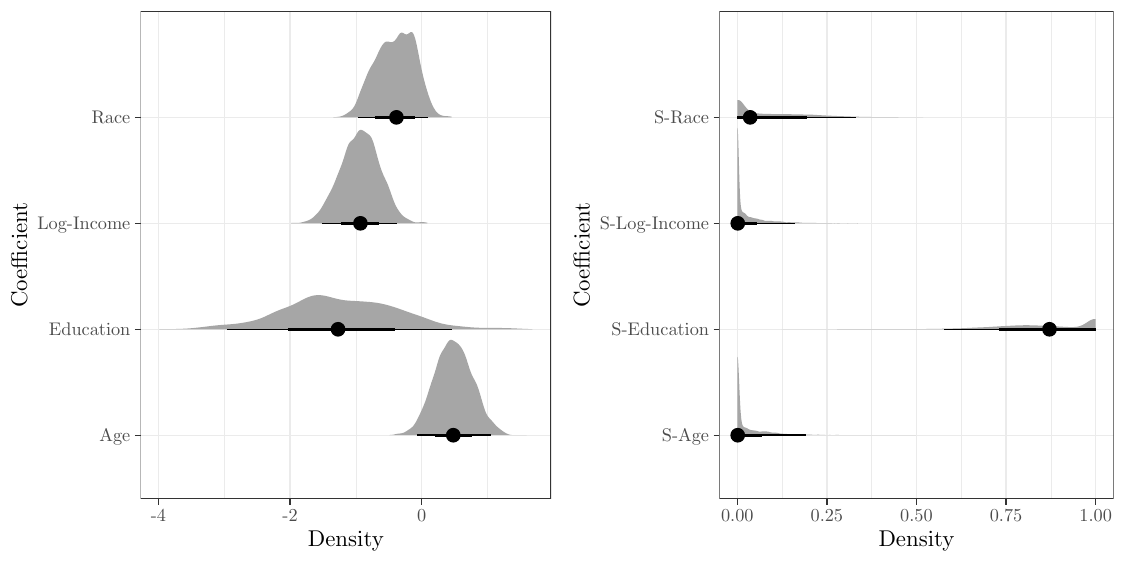}
\par\end{centering}
\caption{Left: Posterior medians, (66\%, 95\%)-credible intervals, and density
estimates for the regression coefficients of the base mode $\beta_{\theta}$.
Right: posterior medians, credible intervals, and density estimates
for the splitting proportions $s_{j}$ for each predictor.}
\label{fig:coefs}

\end{figure}
In Figure \ref{fig:density-est} we display the estimated density
as the level of educational attainment is varied from less-than-high-school
to graduate degree for white women with all other covariates frozen
at their median value. We see that as educational attainment increases
the bulk of the distribution remains concentrated near $0$ (the overall
mean level of BMI) but goes from roughly symmetric to being highly
right-skewed. The nature of this relationship is that, while the modal
value BMI is fairly stable as education level changes, highly educated
women are less likely to be highly obese.

Each predictor $j$ is associated to two coefficients: the base model
coefficient $\beta_{\theta j}$ and the splitting proportion $s_{j}$.
The posterior median, density, and a (66\%, 95\%)-credible interval
is given for each coefficient in Figure \ref{fig:coefs}. Interestingly,
education level is the only relevant predictor in the selection model
$\Phi\{r(y,x)\}$ so that the overall shape of the density is primarily
determined by education. Intuitively, one might expect that education
is only relevant through its indirect effect on income, however our
results suggest this is not the case. Log-income has a strong presence
in the base model as well, while race and age have weaker effects.

\section{Discussion \label{sec:Discussion}}

In this paper we proposed a new method for density regression based
on Bayesian additive regression trees. SBART-DS is suitable for routine
use --- it has a simple default specification, strong theoretical
properties, and can be fit using a tuning-parameter-free Gibbs sampling
algorithm. On simulated data we illustrated how SBART-DS is capable
of filtering out irrelevant variables automatically, giving empirical
support to the theoretical results supporting a faster posterior concentration
rate when the rejection model $\Phi\{r(y,x)\}$ is sparse. Using data
from MEPS we showed how SBART-DS can capture the effect of education
level on the conditional distribution of body mass index.

The general strategy of defining a prior using a rejection sampling
model can be used to extend that SBART-DS model to other domains.
For example, this approach can be extended to survival analysis by
modeling the hazard function $\lambda(y\mid x)$ as the hazard of
a thinned Poisson process $\lambda(y\mid x)=\lambda_{0}(y\mid x)\ \Phi\{r(y,x)\}$.
We will pursue this direction in future work.

\appendix

\section{Proof of Auxiliary Results}
\begin{proof}
[Proof of Lemma \ref{lem:compare}] For posterity, we note that Condition
L implies $\frac{d}{d\mu}\log\Phi(\mu)\le\mathcal{K}$; integrating
both sides on an interval $[L,U]$ gives 
\begin{align}
e^{-\mathcal{K}(U-L)}\le\frac{\Phi(L)}{\Phi(U)}\le\frac{\Phi(U)}{\Phi(L)} & \le e^{\mathcal{K}(U-L)}.\label{eq:supexp}
\end{align}
Let $a=\sqrt{\Phi\{u(y,x)\}}$ and $b=\sqrt{\Phi\{v(y,x)\}}$ and
let $\|a\|_{h}^{2}$ denote the squared $L_{2}$-norm $\int a^{2}(y,x)\ h(y)$
(which implicitly depends on $x$). Then 
\[
H(f_{u},f_{v})=\int\left\Vert \frac{a}{\|a\|_{h}}-\frac{b}{\|b\|_{h}}\right\Vert _{h}\ F_{X}(dx).
\]
Two applications of the triangle inequality gives 
\begin{equation}
H(f_{u},f_{v})\le\int\frac{2\|a-b\|_{h}}{\|a\|_{h}}\ F_{X}(dx)\le2\|1-b/a\|_{\infty}.\label{eq:hellbound}
\end{equation}
The first inequality follows from the triangle inequality while the
second follows from the inequality $\|a-b\|_{h}^{2}=\int a^{2}(y,x)\{1-b(y,x)/a(y,x)\}^{2}\ h(y)\le\|a\|_{h}^{2}\cdot\|1-b/a\|_{\infty}^{2}$.
Next, write $v(y,x)=u(y,x)+\Delta(y,x)$. Applying (\ref{eq:supexp})
and Taylor expanding the function $g(x)=\sqrt{\Phi(\mu+x)/\Phi(\mu)}$
we get 
\begin{align*}
\left|1-\sqrt{\frac{\Phi\{v(y,x)\}}{\Phi\{u(y,x)\}}}\right| & \le|\Delta(y,x)|\frac{\phi\{u(y,x)+\Delta_{1}(y,x)\}}{2\Phi\{u(y,x)+\Delta_{1}(y,x)\}}\ \sqrt{\frac{\Phi\{u(y,x)+\Delta_{1}(y,x)\}}{\Phi\{u(y,x)\}}}\\
 & \le\frac{\mathcal{K}}{2}\|u-v\|_{\infty}\exp\left[\frac{\mathcal{K}\|u-v\|_{\infty}}{2}\right]
\end{align*}
where $\Delta_{1}(y,x)$ is between $0$ and $\Delta(y,x)$. Combining
this with (\ref{eq:hellbound}), $H(f_{u},f_{v})\le\mathcal{K}\|u-v\|_{\infty}\exp(\mathcal{K}\|u-v\|/2)$.
By Lemma 8 of \citet{ghosal2007posterior}, we have 
\begin{align*}
K(f_{u},f_{v}) & \lesssim H^{2}(f_{u},f_{v})\left(1+\log\left\Vert \frac{f_{u}}{f_{v}}\right\Vert _{\infty}\right)\qquad\text{and}\\
V(f_{u},f_{v}) & \lesssim H^{2}(f_{u},f_{v})\left(1+\log\left\Vert \frac{f_{u}}{f_{v}}\right\Vert _{\infty}\right)^{2}.
\end{align*}
Using (\ref{eq:supexp}) we have 
\[
\frac{f_{u}(y\mid x)}{f_{v}(y\mid x)}=\frac{\Phi\{u(y,x)\}\int h(\widetilde{y})\ \frac{\Phi\{v(\widetilde{y},x)\}}{\Phi\{u(\widetilde{y},x)\}}\ \Phi\{u(\widetilde{y},x)\}\ d\widetilde{y}}{\Phi\{v(y,x)\}\int h(\widetilde{y})\ \Phi\{u(\widetilde{y},x)\}\ d\widetilde{y}}\le\exp(2\mathcal{K}\|u-v\|_{\infty}).
\]
Hence $\log\|f_{u}/f_{v}\|_{\infty}\le2\mathcal{K}\|u-v\|_{\infty}$.
\end{proof}
\begin{proof}
[Proof of Proposition \ref{prop:sc}] Consider an extended prior $\widetilde{\Pi}$
on the joint distribution of $(X_{i},Y_{i})$ which places a point
mass at $F_{X}$. We now have that $f_{r}$ is contained in the integrated
Hellinger and Kullback-Leibler neighborhoods whenever $(f_{r},F_{X})$
are in the usual Hellinger and Kullback-Leibler neighborhoods of $F_{0}(dx,dy)=f_{0}(y\mid x)\ dy\ F_{X}(dx)$,
so that the problem reduces to the setting of iid random vectors.
The conditions (a)---(c) match one-to-one with the conditions of
the variant of Theorem 2.1 of \citet{ghosal2000convergence} used
by \citet[page 627]{shen2013adaptive}, and hence suffice to establish
the desired rate of convergence.
\end{proof}
We now prove that Condition L holds for the logit and $t_{\nu}$ links.
\begin{prop}
If $\Phi(\mu)=e^{\mu}/(1+e^{\mu})$ then Condition L holds with $\mathcal{K}=1$.
If $\Phi(\mu)=T_{\nu}(\mu)$ where $T_{\nu}$ is the distribution
function of a $t_{\nu}$ random variable then Condition L holds with
$\mathcal{K}=\sqrt{\nu}$. Conversely, suppose $Z$ is symmetric and
has distribution function $\Phi(\mu)$ and $Z$ is light-tailed in
the sense that for all $\mathcal{K}$ we have $\Pr(Z>z)<e^{-\mathcal{K}z}$
for sufficiently large $z$. Then Condition L fails for $\Phi(\mu)$.
In particular, Condition L fails for the probit link.
\end{prop}
\begin{proof}
For the logistic link it is straight-forward to check that $\phi(\mu)/\Phi(\mu)=1-\Phi(\mu)\le1$.
To prove the result for the $T_{\nu}$ link we begin by deriving a
lower bound for the survival function $\bar{T}_{\nu}(\mu)=\int_{\mu}^{\infty}t_{\nu}(x)\ dx$
for $\mu>0$. Note that the density is $t_{\nu}(\mu)=c/(1+\mu^{2}/\nu)^{p}$
where $p=(\nu+1)/2$ and $c$ is a normalizing constant. Define $z_{0}=1+\mu^{2}/\nu$,
$z=1+x^{2}/\nu$, $\theta_{0}=\arcsin(1/\sqrt{z_{0}})$ and $\theta=\arcsin(1/\sqrt{z})$.
Then after routine substitutions we have 
\begin{align*}
\bar{T}_{\nu}(\mu) & =\int_{z_{0}}^{\infty}\frac{c\sqrt{\nu}}{2\sqrt{z-1}}z^{-p}\ dz=c\sqrt{\nu}\int_{0}^{\theta_{0}}\sin^{2p-2}(\theta)\ d\theta\\
 & \ge c\sqrt{\nu}\int_{0}^{\theta_{0}}\sin^{2p-2}(\theta)\ \cos(\theta)\ d\theta=\frac{c\sqrt{\nu}}{2p-1}z_{0}^{-p+1/2}.
\end{align*}
Using this, the symmetry of the $t_{\nu}$ distribution, and substituting
$2p-1=\nu$, we have
\[
\frac{t_{\nu}(\mu)}{T_{\nu}(\mu)}\le\frac{t_{\nu}(|\mu|)}{\bar{T}_{\nu}(|\mu|)}\le\frac{c}{z_{0}^{p}}\cdot\frac{\sqrt{\nu}}{cz_{0}^{-p+1/2}}=\sqrt{\nu}z_{0}^{-1/2}\le\sqrt{\nu}.
\]
For the converse, set $U=0$ and $L=-z$ in (\ref{eq:supexp}) to
get $\Pr(Z>z)\ge\Phi(0)e^{-\mathcal{K}z}$ for some $\mathcal{K}$.
In particular, $\Pr(Z>z)\ge\exp(-2\mathcal{K}z)$ for large enough
$z$. In the case of the probit link, no such $\mathcal{K}$ can exist
because $\Pr(Z>z)\le e^{-z^{2}/2}$.
\end{proof}
The proof of Theorem \ref{thm:sparse} also requires a tail probability
bound for the number of trees in the ensemble.
\begin{prop}
\label{prop:M-tail} Let $\Pi(M=t)$ satisfy Condition P1. Then there
exist constants $C'_{M1}$ and $C'_{M2}$ such that $\Pi(M\ge t)\le C'_{M1}\exp\{-C'_{M2}t\log t\}$.
\end{prop}
\begin{proof}
By the geometric series formula we have
\[
\Pi(M\ge t)=C_{M1}\sum_{k=t}^{\infty}e^{-C_{M2}k\log k}\le\frac{C_{M1}\exp\{-C_{M2}t\log t\}}{1-t^{-C_{M2}}}=\frac{C_{M1}\exp\{-C_{M2}(t-1)\log t\}}{t^{C_{M2}}-1}.
\]
For $t>2\vee2^{1/C_{M2}}$ this gives $\Pi(M\ge t)\le C_{M1}\exp\{-(C_{M2}/2)t\log t\}$.
The result follows by taking $C'_{M2}=C_{M2}/2$ and $C'_{M1}$ to
be the maximum of $C_{M1}$ and $\exp\{(C_{M2}/2)t\log t\}$ for $t\le2\vee2^{1/C_{M2}}$.
\end{proof}

\section{Proof of Theorem \ref{thm:sparse}}

For completeness, we state two results of \citet{linero2017abayesian}
which will be used in the proof. These two propositions capture the
features of SBART that make it useful in high-dimensional sparse settings
with smooth regression functions.
\begin{prop}
\label{prop:LY-1} Suppose that Condition F and Condition P are satisfied
and $t\ge\alpha(D+1)/(2\alpha+D)$. Then there exist constants $B$
and $C$ independent of $(n,P)$ such that for all sufficiently large
$n$ the prior satisfies 
\[
\Pi(\|r-r_{0}\|\le B\epsilon_{n})\ge e^{-C\epsilon_{n}^{2}},
\]
where $\epsilon_{n}=n^{-\alpha/(2\alpha+D)}\log(n)^{t}+\sqrt{D\log(P+1)/n}$.
\end{prop}
\begin{proof}
This is implied by Theorem 2 of \citet{linero2017abayesian}; the
only modification required is that Condition P1 and Condition P2 are
modified from \citet{linero2017abayesian}, but these modifications
do not change the proof strategy.
\end{proof}
\begin{prop}
\label{prop:LY-2} For fixed positive constants $\epsilon,\sigma_{1},\sigma_{2},T,A$
and integers $n,H,d$ define the set 
\begin{align*}
\mathcal{G}= & \Big\{ f(\cdot)=\sum_{t=1}^{T}g(\cdot;\Tree_{t},\sM_{t}):\ T\le An\epsilon^{2},\text{ each tree has depth at most }H,\\
 & \text{the common bandwidth parameter \ensuremath{\tau} satisfies }\sigma_{1}\le\tau^{-1}\le\sigma_{2},\\
 & \text{the total number of splitting directions is at most }d\text{ out of \ensuremath{P+1}},\\
 & \text{for each \ensuremath{(t,\ell),} }\mu_{t\ell}\in[-U,U]\Big\}.
\end{align*}
Then there exists a constant $C_{\psi}$ depending only the gating
function $\psi$ of the SBART prior satisfying Condition P such that
the following holds:
\begin{enumerate}
\item Covering entropy control: $\log N(\mathcal{G},C_{\psi}\varepsilon,\|\cdot\|_{\infty})\leq d\log(P+1)+3An\varepsilon^{2}\,2^{H}\log\big(d\,\sigma_{1}^{-1}\sigma_{2}^{2}An\varepsilon\,2^{H}U\big)$;
and
\item Complement probability bound: if $H\geq d_{0}$, then $\Pi(\mathcal{G}^{c})\leq C'_{M1}\exp\{-C'_{M2}\,An\varepsilon^{2}\log(An\epsilon^{2})\}+2^{H}\,An\varepsilon^{2}$
$\cdot\big[\exp\{-E\,d\log(P+1)\}+C_{\mu1}\,\exp\{-U^{C_{\mu2}}\}\big]+C_{\tau1}\,\exp\{-\sigma_{1}^{-C_{\tau2}}\}+C_{\tau3}\exp\{-\sigma_{2}^{C_{\tau4}}\}$
for some constant $E>0$ depending only on hyperparameter $\xi>1$
in the Dirichlet prior.
\end{enumerate}
\end{prop}
\begin{proof}
The proof is the same as the proof of Lemma 1 of the supplementary
material of \citet{linero2017abayesian}, except that Proposition
\ref{prop:M-tail} is used in the complementary probability bound.
\end{proof}
\begin{proof}
[Proof of Theorem \ref{thm:sparse}] Let $B$ and $C$ be chosen as
in Proposition \ref{prop:LY-1}. By Lemma \ref{lem:compare}, note
that for sufficiently large $n$, we have $\{f_{r}:\|r-r_{0}\|\le B\epsilon_{n}\}\subseteq K(BC_{\mathcal{K}}\epsilon_{n})$
and $\{f_{r}:\|r-r_{0}\|\le\epsilon_{n}\}\subseteq\{f_{r}:H(f_{r},f_{r_{0}})\le C_{\mathcal{K}}\epsilon_{n}\}$
where $C_{\mathcal{K}}$ is a constant depending only on $\mathcal{K}$.
Hence 
\[
\Pi\{f\in K(BC_{\mathcal{K}}\epsilon_{n})\}\ge e^{-Cn\epsilon_{n}^{2}}.
\]
To lighten notation, we redefine $\epsilon_{n}$ throughout the rest
of the proof to be $\epsilon_{n}BC_{\mathcal{K}}$ and $C$ to be
$C/(BC_{\mathcal{K}})^{2}$ so that we have $\Pi\{f\in K(\epsilon_{n})\}\ge e^{-Cn\epsilon_{n}^{2}}$.
This verifies (c) of Proposition \ref{prop:sc} using the modified
choice of $\epsilon_{n}$.

Next, for a large constant $\kappa$ to be chosen later, set $A=\kappa/\log n$,
$\sigma_{1}^{-C_{\tau2}}=\sigma_{2}^{C_{\tau4}}=U^{C_{\mu2}}=\kappa n\epsilon_{n}^{2}$,
$H=d_{0}$, and $d=\lfloor\kappa n\epsilon_{n}^{2}/\log(P+1)\rfloor$
for the set $\mathcal{G}_{n}$ in Proposition \ref{prop:LY-2}. Plugging
these constants into the covering entropy bound, for sufficiently
large $n$ this implies that for $p_{1}=2+C_{\mu2}^{-1}+C_{\tau2}^{-1}+2C_{\tau4}^{-1}$
and $p_{2}=2p_{1}-1$ we have
\[
\log N(\mathcal{G}_{n},C_{\psi}\epsilon_{n},\|\cdot\|_{\infty})\le\kappa n\epsilon_{n}^{2}\left\{ 1+\frac{3\cdot2^{d_{0}}}{\log n}\log\left(\frac{2^{d_{0}}(\kappa n)^{p_{1}}\epsilon_{n}^{p_{2}}}{\log n}\right)\right\} \le\kappa'n\epsilon_{n}^{2}
\]
for some $\kappa'$ larger than $\kappa$ depending only on $\kappa$
and the constants in Condition P. Define $\mathcal{F}_{n}=\{f_{r}:r\in\mathcal{G}_{n}\}$.
By Lemma \ref{lem:compare}, for large enough $n$ any $C_{\psi}\epsilon_{n}$-net
$\mathcal{G}_{\text{net}}$ for $\mathcal{G}_{n}$ can be converted
into a $C_{\mathcal{K}}C_{\psi}\epsilon_{n}$-net $\Sieve_{\text{net}}=\{f_{r}:r\in\mathcal{G}_{\text{net}}\}$
for $\mathcal{F}_{n}$. Hence we also have the bound
\[
\log N(\mathcal{F}_{n},C_{\mathcal{K}}C_{\psi}\epsilon_{n},H)\le\log N(\mathcal{G}_{n},C_{\psi}\epsilon_{n},\|\cdot\|_{\infty})\le\kappa'n\epsilon_{n}^{2}
\]
which establishes condition (a) of Proposition \ref{prop:sc} with
$\bar{\epsilon}_{n}=C_{\mathcal{K}}C_{\psi}\epsilon_{n}$ and $C_{N}=\kappa'/(C_{\mathcal{K}}C_{\psi})^{2}$.
Finally, we show condition (b) holds. Applying the complementary probability
bound we can make $\Pi(\mathcal{F}_{n}^{c})\le\exp\{-(C+4)n\epsilon_{n}^{2}\}$
for any choice of $C$ by taking $\kappa$ sufficiently large. To
see why, note for example that $n\epsilon_{n}^{2}\ge an^{b}$ for
some positive constants $(a,b)$ so that for large $n$ we have
\[
An\epsilon_{n}^{2}\log(An\epsilon_{n}^{2})\ge\kappa n\epsilon_{n}^{2}\left\{ \frac{\log(\kappa/\log n)+\log a}{\log n}+b\right\} \ge\frac{\kappa b}{2}n\epsilon_{n}^{2}.
\]
Using similar arguments, for large $n$ we can bound each term of
the complementary probability bound by $\exp\{-\kappa\delta n\epsilon_{n}^{2}/2\}$
for some $\delta$ depending only on the constants in Condition P.
Taking $\kappa$ sufficiently large we can make the total bound less
than $\exp\{-(C+4)n\epsilon_{n}^{2}\}$ for arbitrary $C$. This proves
condition (b).
\end{proof}
\bibliographystyle{apalike}
\bibliography{mybib}

\end{document}